%% file: TwoPartyCoinFlipToKA.tex
\documentclass[11pt]{article}
\usepackage{amsfonts,amsmath,amssymb,amsthm,boxedminipage,color,url,fullpage}
\usepackage{enumerate,paralist}
\usepackage[numbers]{natbib} 
\usepackage[pdfstartview=FitH,colorlinks,linkcolor=blue,filecolor=blue,citecolor=blue,urlcolor=blue]{hyperref} 
\usepackage[labelfont=bf]{caption}
\usepackage{aliascnt}
\usepackage{cleveref}
\usepackage{xspace}
\usepackage{relsize}
\usepackage{graphics}

\usepackage{float}
\floatstyle{boxed}
\restylefloat{figure}

\renewcommand{\pi}{ \Pi }

\input{macros}

\renewcommand{\Cc}{\mathsf{D}}
\renewcommand{\I}{\mathcal{K}}

\newcommand{\tpprot}[2]{\left( #1,#2\right)}
\newcommand{\tpi}{\wt{\pi}}
\newcommand{\tAc}{\wt{\Ac}}
\newcommand{\tBc}{\wt{\Bc}}

\newcommand{\np}{n}

\newcommand{\di}{\mathsf{D}}

\newcommand{\wh}[1]{\widehat{#1}}
\newcommand{\hPi}{\widehat{\pi}}
\newcommand{\hAc}{\widehat{\Ac}}
\newcommand{\hBc}{\widehat{\Bc}}

\newcommand{\Dcr}{\MathAlgX{Dcr}}

\newcommand{\hDcr}{\Dcr}

\newcommand{\tPi}{\widetilde{\pi}}
\newcommand{\htPi}{\widehat{\pi}}

\newcommand{\Ec}{\MathAlgX{Eve}}
\newcommand{\secParam}{\kappa}

\newcommand{\pk}{\secParam}
\newcommand{\Po}{\overline{\Pc}}

\newcommand{\Inote}[1]{\authnote{Iftach}{#1}}
\newcommand{\Nnote}[1]{\authnote{Nikos}{#1}}

\title{On the Complexity of Fair Coin Flipping\thanks{A preliminary version of the work appeared in~\cite{HMO18}.}\Draft{\\{\small \sc Working Draft: Please Do Not Distribute}}}
 \author{Iftach Haitner\thanks{School of Computer Science, Tel Aviv University. E-mail: \texttt{iftachh@cs.tau.ac.il}. Member of the  Check Point Institute for Information Security.} %
 	\footnote{Research supported by ERC starting grant 638121.}
 	\and Nikolaos Makriyannis\thanks{School of Computer Science, Tel Aviv University. E-mail: \texttt{n.makriyannis@gmail.com}.}~\footnotemark[3]
 	\and Eran Omri\thanks{Department of Computer Science, Ariel University. E-mail: \texttt{omrier@ariel.ac.il}. Research supported by ISF grant 152/17, and by the Ariel Cyber Innovation Center in conjunction with the Israel National Cyber directorate in the Prime Minister's Office.}
 }

\begin{document}
	\sloppy
	  \maketitle

\begin{abstract}
A two-party coin-flipping protocol is $\eps$-fair if no efficient adversary can bias the output of the honest party (who always outputs a bit, even if the other party aborts) by more than $\eps$. \citeauthor{Cleve86} [STOC '86] showed that $r$-round $o(1/r)$-fair coin-flipping protocols do not exist.  \citeauthor*{AwerbuchBCGM1985} [Manuscript '85] constructed a $\Theta(1/\sqrt{r})$-fair coin-flipping protocol, assuming the existence of one-way functions. \citeauthor*{MoranNS16} [Journal of Cryptology '16] constructed an $r$-round coin-flipping protocol that is $\Theta(1/r)$-fair (thus matching the aforementioned lower bound of \citeauthor{Cleve86} [STOC '86]), assuming the existence of oblivious transfer.

The above gives rise to the intriguing question of whether oblivious transfer, or more generally ``public-key primitives,'' is required for an $o(1/\sqrt r)$-fair coin flipping protocol.  Towards answering this intriguing question, \citet*{MW20} [Crypto '18] have recently showed that in the \textit{random oracle model (ROM)}, any coin-flipping protocol can be biased by $\Omega(1/\sqrt{r})$. This  implies that $o(1/\sqrt r)$-fair coin-flipping protocol cannot be constructed from one-way function, or from a family of collision-resistant hash functions, in a \textit{black-box} way.  This result does not rule out, however, non black-box constructions, and black-box constructions based on primitives that cannot be realized in the ROM.

We make a different progress towards answering  above question by showing that, for any constant $r\in \N$, the existence of an $1/(c\cdot \sqrt{r})$-fair, $r$-round coin-flipping protocol implies the existence of an infinitely-often key-agreement protocol, where $c$ denotes some universal constant (independent of $r$). 
Our reduction is \emph{non} black-box and makes a novel use of the recent dichotomy for two-party protocols of \citeauthor*{HNOSS18New} to facilitate a two-party variant of the recent attack of \citeauthor*{BHMO18New} on multi-party coin-flipping protocols. 
\end{abstract}

\Tableofcontents

\input{Intro}

\input{Prelims}

\input{ReductionToKA}

\bibliographystyle{abbrvnat}
\bibliography{crypto}
\end{document}

%% file: macros.tex
\newcommand{\remove}[1]{}
\newcommand{\Draft}[1]{\ifdefined\IsDraft\texttt{ #1} \fi}

\newcommand{\TLLNCS}[2]{\ifdefined\IsLLNCS#1\else #2 \fi}

\ifdefined\IsDraft
\newcommand{\authnote}[2]{{\bf [{\color{red} #1's Note:} {\color{blue} #2}]}}
\else
\newcommand{\authnote}[2]{}
\fi

\ifdefined\IsDraft
\newcommand{\changed}[1]{{\color{blue} #1}}
\else
\newcommand{\changed}[1]{#1}
\fi

\ifdefined\IsDraft
\newcommand{\deleted}[1]{{\color{blue} ~Deleted:~{\color{red} #1}}}
\else
\newcommand{\deleted}[1]{}
\fi

\newcommand{\sdotfill}{\textcolor[rgb]{0.8,0.8,0.8}{\dotfill}} 

\newenvironment{protocol}{\begin{proto}}{\vspace{-\topsep}\sdotfill\end{proto}}
\newenvironment{algorithm}{\begin{algo}}{\vspace{-\topsep}\sdotfill\end{algo}}

\newcommand{\Ensuremath}[1]{\ensuremath{#1}\xspace}
\newcommand{\MathAlg}[1]{\mathsf{#1}}

\newcommand{\MathAlgX}[1]{\Ensuremath{\MathAlg{#1}}}

\newcommand{\ie}  {i.e.,\xspace}
\newcommand{\eg}  {e.g.,\xspace}

\newcommand{\wrt} {with respect to\xspace}
\newcommand{\wlg} {without loss of generality\xspace}

\newcommand{\ceil}[1]{\left\lceil #1 \right\rceil}

\newcommand{\set}[1]{\ens{#1}}

\newcommand{\ith}{\Ensuremath{i^{\rm th}}}

\newcommand{\R}{{\mathbb R}}
\newcommand{\N}{{\mathbb{N}}}

\newcommand{\zo}{\{0,1\}}

\newcommand{\eps}{\varepsilon}

\newcommand{\la}{\gets}

\newcommand{\poly}{\operatorname{poly}}
\newcommand{\polylog}{\operatorname{polylog}}

\newcommand{\negl}{\operatorname{neg}}

\newcommand{\Supp}{\operatorname{Supp}}

\renewcommand{\cref}{\Cref}

\TLLNCS{
	\newaliascnt{claiml}{theorem}
	\newtheorem{claiml}[claiml]{Claim}
	\aliascntresetthe{claiml}
	
	\renewenvironment{claim}{\begin{claiml}}{\end{claiml}}
}
{
	\newtheorem{theorem}{Theorem}[section]

	\newaliascnt{lemma}{theorem}
	
	\aliascntresetthe{lemma}
	
	\newaliascnt{claim}{theorem}
	\newtheorem{claim}[claim]{Claim}
	\aliascntresetthe{claim}

	\newaliascnt{corollary}{theorem}
	
	\aliascntresetthe{corollary}
	
	\newaliascnt{proposition}{theorem}
	
	\aliascntresetthe{proposition}

	\newaliascnt{conjecture}{theorem}
	
	\aliascntresetthe{conjecture}

	\newaliascnt{definition}{theorem}
	\newtheorem{definition}[definition]{Definition}
	\aliascntresetthe{definition}

	\newaliascnt{remark}{theorem}
	\newtheorem{remark}[remark]{Remark}
	\aliascntresetthe{remark}

	\newaliascnt{example}{theorem}
	
	\aliascntresetthe{example}
}

\crefname{lemma}{Lemma}{Lemmas}
\crefname{figure}{Figure}{Figures}
\crefname{claim}{Claim}{Claims}
\crefname{corollary}{Corollary}{Corollaries}
\crefname{proposition}{Proposition}{Propositions}
\crefname{conjecture}{Conjecture}{Conjectures}
\crefname{definition}{Definition}{Definitions}
\crefname{remark}{Remark}{Remarks}
\crefname{exmaple}{Example}{Examples}

\newaliascnt{construction}{theorem}

\aliascntresetthe{construction}
\crefname{construction}{Construction}{Constructions}

\newaliascnt{fact}{theorem}

\aliascntresetthe{fact}
\crefname{fact}{Fact}{Facts}

\newaliascnt{notation}{theorem}

\aliascntresetthe{notation}
\crefname{notation}{Notation}{Notation}

\crefname{equation}{Equation}{Equations}

\newaliascnt{proto}{theorem}

\newtheorem{proto}[proto]{Protocol}

\aliascntresetthe{proto}
\crefname{proto}{protocol}{protocols}

\newaliascnt{algo}{theorem}
\newtheorem{algo}[algo]{Algorithm}
\aliascntresetthe{algo}
\crefname{algo}{algorithm}{algorithms}

\newaliascnt{expr}{theorem}
\newtheorem{expr}[expr]{Experiment}
\aliascntresetthe{expr}
\crefname{experiment}{experiment}{experiments}

\def\FullBox{$\Box$}
\def\qed{\ifmmode\qquad\FullBox\else{\unskip\nobreak\hfil
		\penalty50\hskip1em\null\nobreak\hfil\FullBox
		\parfillskip=0pt\finalhyphendemerits=0\endgraf}\fi}

\def\qedsketch{\ifmmode\Box\else{\unskip\nobreak\hfil
		\penalty50\hskip1em\null\nobreak\hfil$\Box$
		\parfillskip=0pt\finalhyphendemerits=0\endgraf}\fi}

\newcommand{\ex}[1]{\Ex\left[#1\right]}
\newcommand{\Ex}{{\mathbf E}}
\renewcommand{\Pr}{{\mathrm {Pr}}}
\newcommand{\pr}[1]{\Pr\left[#1\right]}

\newcommand{\Ac}{\MathAlgX{A}}
\newcommand{\Gc}{\MathAlgX{G}}
\newcommand{\As}{\MathAlgX{\Ac^\ast}}
\newcommand{\Dc}{\MathAlgX{D}}
\newcommand{\Fc}{\MathAlgX{F}}
\newcommand{\Pc}{\mathsf{P}}

\newcommand{\Bc}{\mathsf{B}}
\newcommand{\Cc}{\mathsf{C}}

\newcommand{\ens}[1]{\left\{#1\right\}}
\newcommand{\size}[1]{\left|#1\right|}

\newcommand{\cindist}{\mathbin{\stackrel{\rm c}{\approx}}}

\newcommand{\Uni}{{\mathord{\mathcal{U}}}}

\newcommand{\prob}[1]{\mathsf{\textsc{#1}}}

\newcommand{\SD}{\prob{SD}}

\newcommand{\I}{\mathcal{I}}

\newcommand{\ppt}{{\sc ppt}\xspace}
\newcommand{\pptm}{{\sc pptm}\xspace}

\newcommand{\cD}{\mathcal{D}}

\newcommand{\st}{\text{ s.t.\ }}

\newcommand{\D}{{{\mathsf{D}}}}

\newcommand{\cs}{{\cal{S}}}

\theoremstyle{definition}

\theoremstyle{plain}

\theoremstyle{remark}

\newcommand{\supp}{\mathrm{supp}}

\newcommand{\wt}[1]{\widetilde{#1}}

\newcommand{\wb}[1]{\overline{#1}}

\newcommand{\Tableofcontents}{
	\ifdefined\IsLLNCS \else
	\thispagestyle{empty}
	\pagenumbering{gobble}
	\clearpage
	\ifdefined\IsSubmission \else
	\setcounter{tocdepth}{2}
	\tableofcontents
	\thispagestyle{empty}
	\clearpage
	\fi
	\pagenumbering{arabic}
	\fi
}

%% file: Intro.tex
\section{Introduction}\label{sec:intro}
In a two-party coin flipping protocol, introduced by \citet{Blum83}, the parties wish to output a common (close to) uniform bit, even though one of the parties may be corrupted and try to bias the output. Slightly more formally, an $\eps$-fair coin flipping protocol should satisfy the following two properties: first, when both parties behave honestly (\ie follow the prescribed protocol), they both output the \emph{same} uniform bit. Second, in the presence of a corrupted party that may deviate from the protocol arbitrarily, the distribution of the honest party's output may deviate from the uniform  distribution (unbiased bit) by at most $\eps$. We emphasize that the above notion requires an honest party to \emph{always} output a bit, regardless of what the corrupted party does, and, in particular, it is not allowed to abort if a cheat is detected.\footnote{Such protocols are typically addressed as having \emph{guaranteed output delivery}, or, abusing terminology, as \emph{fair}.}~Coin flipping is a fundamental primitive with numerous applications, and thus lower bounds on coin flipping protocols yield analogous bounds for many basic cryptographic primitives, including other inputless primitives and secure computation of functions that take input (\eg XOR). 

 In his seminal work, \citet{Cleve86} showed that, for \emph{any} efficient two-party $r$-round coin flipping protocol, there exists an efficient adversarial strategy that biases the output of the honest party by $\Theta(1/r)$. The above lower bound on coin flipping protocols was met for the two-party case by \citet*{MoranNS16} improving over the $\Theta(\np/\sqrt{r})$-fairness achieved by the majority protocol of \citet*{AwerbuchBCGM1985}. The protocol of \cite{MoranNS16}, however, uses oblivious transfer; to be compared with the protocol of \cite{AwerbuchBCGM1985} that can be based on any one-way function. An intriguing open question is whether oblivious transfer, or more generally ``public-key primitives,'' is required for an $o(1/\sqrt r)$-fair coin flip. The question was partially answered in the black-box setting by \citet*{Dachman11} and \citet*{DachmanMM14}, who showed that \emph{restricted} types of fully black-box reductions cannot establish $o(1/\sqrt r)$-bias coin flipping protocols from one-way functions. In particular, for constant-round coin flipping protocols, \cite{Dachman11} yields that black-box techniques from one-way functions can only guarantee fairness of order $1/\sqrt{r}$.

\subsection{Our Results}
Our main result is that constant-round coin flipping protocols with better bias compared to the majority protocol of \cite{AwerbuchBCGM1985} imply the existence of infinitely-often key-agreement. We recall that infinitely-often key-agreement protocols satisfy correctness (parties agree on a common bit with overwhelming probability), and, for an infinite number of security parameters, no efficient eavesdropper can deduce the output with probability noticeably far from a random guess.\footnote{While infinitely-often key-agreement protocols are useless from a cryptographic point of view \changed{(as they do not guarantee security for every security parameter), constructing such protocols appears to be as hard as obtaining full-fledged key agreement protocols.}} 
\begin{theorem}[Main result, informal]\label{thm:Infmain}
	For any (constant) $r\in \N$, the existence of an $1/(c\cdot \sqrt r)$-fair, $r$-round coin flipping protocol implies the existence an infinitely-often key-agreement protocol, for $c>0$ being a universal constant (independent of $r$).
\end{theorem}

As in \cite{Cleve86,Dachman11,DachmanMM14}, our result extends via a simple reduction to general multi-party coin flipping protocols (with more than two-parties) without an honest majority. 
Our non black-box  reduction  makes a novel use of the recent dichotomy for two-party protocols of \citet{HNOSS18New}. Specifically, assuming that io-key-agreement does not exist and applying \citeauthor{HNOSS18New}'s dichotomy, we show that a two-party variant of the recent multi-party attack of \citet{BHMO18New} yields a $ 1/(c\cdot \sqrt{r}) $-bias attack.

 \subsection{Our Technique}\label{sec:intro:Technique}
 Let $\pi =(\Ac,\Bc)$ be a $r$-round two-party coin flipping protocol. We show that the nonexistence of key-agreement protocols yields an efficient $\Theta(1/\sqrt r)$-bias attack on $\pi$. We start by describing the $1/\sqrt{r}$-bias \emph{inefficient} attack of \citet{CleveI93}, and the approach of \citet{BHMO18New} towards making this attack efficient. We then explain how to use the recent results by \citet{HNOSS18New} to obtain an efficient attack (assuming the nonexistence of io-key-agreement protocols).

\subsubsection{\citeauthor{CleveI93}'s Inefficient Attack} \label{sec:intro:CI}
We describe the inefficient $1/\sqrt{r}$-bias attack due to \citet{CleveI93}. Let $M_1,\ldots,M_r$ denote the messages in a random execution of $\pi$, and let $C$ denote the (\wlg) always common output of the parties in a random honest execution of $\pi$. Let $X_i = \ex{C \mid M_{\le i} }$. Namely, $M_{\le i} = M_1,\ldots,M_i$ denotes the partial transcript of $\pi$ up to and including round $i$, and $X_i$ is the expected outcome of the parties in $\pi$ given $M_{\le i}$. It is easy to see that $X_0,\ldots,X_r$ is a martingale sequence: $\ex{X_i \mid X_{0},\ldots, X_{i-1}} = X_{i-1}$ for every $i$. Since the parties in an honest execution of $\pi$ output a uniform bit, it holds that $X_0 = \pr{C = 1} = 1/2$ and $X_r\in \zo$. \citet{CleveI93} (see \citet{BHMO18New} for an alternative simpler proof) prove that, for such a sequence (omitting absolute values and constant factors),
\begin{align}\label{eq:CI:Jump}
&\mbox{Gap:}&\hfill \pr{\exists i\in [r] \colon X_i- X_{i-1} \ge 1/\sqrt{r}} \ge 1/2
\end{align}

 Let the \ith \emph{backup value} of party $\Pc$, denoted $Z_i^\Pc$, be the output of party $\Pc$ if the other party aborts prematurely \emph{after} the \ith message was sent (recall that the honest party must always output a bit, by definition). In particular, $Z^\Pc_r$ denotes the final output of $\Pc$ (if no abort occurred).  We claim that \wlg for both $\Pc\in \set{\Ac, \Bc}$ it holds that
 \begin{align}\label{eq:CI:BU}
 &\mbox{Backup values approximate outcome:}& \pr{\exists i\in [r] \colon \size{X_{i} - \ex{Z_{i}^\Pc \mid M_{\le i} }} \ge 1/2\sqrt{r}} \le 1/4
 \end{align}
  To see why, assume \cref{eq:CI:BU} does not hold.  Then,  the (possibly inefficient) adversary controlling $\wb{\Pc}\in \set{\Ac,\Bc}\setminus \Pc$  that aborts at the end of round $i$ if $(-1)^{1-z}\cdot (X_{i} - \ex{Z_{i}^\Pc \mid M_{\le i} }) \ge 1/\sqrt{r}$,  for suitable   $z\in \zo$,  biases the output of $\Pc$ towards $1-z$ by $\Theta(1/\sqrt r)$.
  
   Finally, since the coins of the parties are \emph{independent} conditioned on the transcript (a fundamental fact about protocols), if party $\Ac$ sends the $(i+1)$ message then
 \begin{align}\label{eq:CI:Ind}
 &\mbox{Independence:}& \ex{Z_{i}^\Bc \mid M_{\le i} } = \ex{Z_{i}^\Bc \mid M_{ \le i+1} }
 \end{align}
 
 Combining the above observations yields that \wlg:
\begin{align}\label{eq:CIGapZ}
\pr{\exists i\in [r] \colon \text{$\Ac$ sends the \ith message} \land X_{i} - \ex{Z_{i-1}^\Bc \mid M_{\le i} } \ge 1/2\sqrt{r}} \ge 1/8
\end{align}
\cref{eq:CIGapZ} yields the following (possibly inefficient) attack for a corrupted party $\Ac$ biasing $\Bc$'s output towards zero: before sending the \ith message $M_i$, party $\Ac$ aborts if $ X_{i} - \ex{Z_{i-1}^\Bc \mid M_{\le i} } \ge 1/2\sqrt{r}$. By \cref{eq:CIGapZ}, this attack biases $\Bc$'s output towards zero by $\Omega(1/2\sqrt{r})$.

The clear limitation of the above attack is that, assuming one-way functions exist, the value of $X_i=\ex{C \mid M_{\le i} = (m_1,\ldots, m_i) }$ and of $\ex{Z_{i}^\Pc \mid M_{\le i} = (m_1,\ldots, m_i) } $ might \emph{not} be efficiently computable as a function of $t$.\footnote{For instance, the first two messages might contain commitments to the parties' randomness.} Facing this difficulty, \citet{BHMO18New} considered the martingale sequence $X_i = \ex{C \mid Z^\Pc_{\le i} }$ (recall that $Z_i^\Pc$ is the \ith backup value of $\Pc$). It follows that, for constant-round protocols, the value of $X_i$ is only a function of a constant size string, and thus it is efficiently computable (\cite{BHMO18New} have facilitated this approach for protocols of super-constant round complexity, see \cref{fn:intro:1}). The price of using the alternative sequence $X_1,\ldots,X_r$ is that the independence property (\cref{eq:CI:Ind}) might no longer hold. Yet, \cite{BHMO18New} manage to facilitate the above approach into an efficient $\wt{\Omega}(1/\sqrt r)$-attack on \emph{multi-party} protocols. In the following, we show how to use the dichotomy of \citet{HNOSS18New} to facilitate a two-party variant of the attack from \cite{BHMO18New}.

 \subsubsection{Nonexistence of Key-Agreement Implies an Efficient Attack}

 Let $U_p$ denote the Bernoulli random variable taking the value $1$ with probability $p$, and let $P \cindist_\rho Q$ stand for $Q$ and $P$ are $\rho$-computationally indistinguishablity (\ie an efficient distinguisher cannot tell $P$ from $Q$ with advantage better than $\rho$). We are using two results by \citet{HNOSS18New}. The first one given below holds for any two-party protocol. 
 
 \begin{theorem}[\citet{HNOSS18New}'s forecaster, informal]\label{thm:intro:ForcasterInf}
 	Let $\Delta = \tpprot{\Ac}{\Bc}$ be a single-bit output (each party outputs a bit) two-party protocol. Then, for any constant $\rho>0$, there exists a constant output-length poly-time algorithm (\textit{forecaster}) \Fc mapping transcripts of $\Delta$ into (the binary description of) pairs in $[0,1] \times [0,1]$ such that the following holds: let $(X,Y,T)$ be the parties outputs and transcript in a random execution of $\Delta$\,, then
 	
 	\begin{itemize}
 		\item 	$(X,T) \cindist_\rho (U_{p^{\Ac}},T)_{(p^{\Ac},\cdot) \la \Fc(T)}$, and 
 		
 		\item $(Y,T) \cindist_\rho (U_{p^{\Bc}},T)_{(\cdot, p^{\Bc}) \la \Fc(T)}$.
 	\end{itemize}
 
 \end{theorem}
Namely, given the transcript, $\Fc$ forecasts the output-distribution for each party in a way that is computationally indistinguishable from (the distribution of) the real output. 


Consider the $(r+1)$-round protocol $\tPi = (\tAc,\tBc)$, defined by $\tAc$ sending a random $i\in [r]$ to $\tBc$ as the first message and then the parties interact in a random execution of $\pi$ for the first $i$ rounds. At the end of the execution, the parties output their \ith backup values $z_i^{\Ac}$ and $z_i^{\Bc}$ and halt. Let $\Fc$ be the forecaster for $\tPi$ guaranteed by \cref{thm:intro:ForcasterInf} for $\rho= 1/r^2$ (note that $\rho$ is indeed constant). A simple averaging argument yields that
\begin{align} \label{eq:indistipi}
(Z_i^{\Pc},M_{\le i}) \cindist_{1/r} (U_{p^{\Pc}},M_{\le i})_{(p^{\Ac},p^{\Bc}) \la \Fc(M_{\le i})}
\end{align}
for both $\Pc \in \set{\Ac,\Bc}$ and every $i\in[r]$, letting $\Fc(m_{\le i})= \Fc(i,m_{\le i})$. Namely, $\Fc$ is a good forecaster for the partial transcripts of $\pi$.

 Let $M_1,\ldots,M_r$ denote the messages in a random execution of $\pi$ and let $C$ denote the output of the parties in $\pi$. Let $F_i = \tpprot{F_i^\Ac}{F_i^\Bc}=\Fc(M_{\le i})$ and let $X_i = \ex{C \mid F_{\le i} }$. It is easy to see that $X_1,\ldots,X_r$ is a martingale sequence and that $X_0 = 1/2$. We assume \wlg that the last message of $\pi$ contains the common output. Thus, it follows from \cref{eq:indistipi} that $F_r\approx (C,C) \in \set{(0,0),(1,1)}$ (otherwise, it will be very easy to distinguish the forecasted outputs from the real ones, given $M_r$). 
Hence, similarly to \cref{sec:intro:CI}, it holds that
 \begin{align}\label{eq:Dec:Jump}
&\mbox{Gap:}&\hfill \pr{\exists i\in [r] \colon X_i- X_{i-1} \ge 1/\sqrt{r}} \ge 1/2
\end{align}
 Since $F_i$ has constant-size support and since $\pi$ is constant round, it follows that $X_i$ is efficiently computable from $M_{\le i}$.\footnote{\label{fn:intro:1}In the spirit of \citet{BHMO18New}, we could have modified the definition of the $X_i$'s to make them efficiently computable even for non constant-round protocols. The idea is to define $X_i= \ex{C \mid F_i,X_{i-1}}$. While the resulting sequence might not be a martingale, \cite{BHMO18New} proves that a $1/\sqrt{r}$-gap also occurs with constant probability \wrt such a sequence. Unfortunately, we cannot benefit from this improvement, since the results of \citet{HNOSS18New} only guarantees indistinguishablity for constant $\rho$, which makes it useful only for attacking constant-round protocols.}
 
Let $Z_i^\Pc$ denote the backup value computed by party $\Pc$ in round $i$ of a random execution of $\pi$. The indistinguishablity of \Fc yields that $\ex{Z_{i}^\Pc \mid F_{\le i} } \approx F_i^\Pc$. Similarly to \cref{sec:intro:CI}, unless there is a simple $1/\sqrt{r}$-attack, it holds that 
 \begin{align}\label{eq:Dec:BU}
 &\mbox{Backup values approximate outcome:}& \pr{\exists i\in [r] \colon \size{X_{i} - \ex{Z_{i}^\Pc \mid F_{\le i} }} \ge 1/2\sqrt{r}} \le 1/4
 \end{align}
 Thus, for an efficient variant of \cite{CleveI93}'s attack, it suffices to show that 
  \begin{align}\label{eq:Dec:Ind}
 &\mbox{Independence:}& \ex{Z_{i}^\Pc \mid F_{ \le i}} \approx \ex{Z_{i}^\Pc \mid F_{ \le i+1} }
 \end{align}
 for every $\Pc \in \set{\Ac,\Bc}$ and round $i$ in which party $\Po\in \set{\Ac,\Bc}\setminus \set{\Pc}$ sends the $(i+1)$ message. However, unlike \cref{eq:CI:Ind} in \cref{sec:intro:CI}, \cref{eq:Dec:Ind} does not hold unconditionally (in fact, assuming oblivious transfer exists, the implied attack must fail for some protocols, yielding that \cref{eq:Dec:Ind} is false for these protocols). Rather, we relate \cref{eq:Dec:Ind} to the existence of a key-agreement protocol. Specifically, we show that if \cref{eq:Dec:Ind} is not true, then there exists a key-agreement protocol. 

\paragraph{Proving that $F_{i+1}$ and $Z_{i}^\Pc$ are approximately independent given $F_{\le i}$.}

 The next (and last) argument is the most technically challenging part of our proof. At this time, we provide a brief yet meaningful overview of the technique. The full details are provided in the main body (\cref{claim:AttackOpurunity} in  \cref{sec:reduction}). 
 
We show that  assuming nonexistence of io-key-agreement, $F_{i+1}$ and $Z_{i}^\Pc$ are approximately independent given $F_{\le i}$.  In more detail,  the triple $(Z_{i}^\Pc,F_{i+1}, F_{\le i})$ is $\rho$-indistinguishable from $(Y_1,Y_2, F_{\le i})$ where $(Y_1,Y_2)$ is a pair of random variables that are mutually independent given $F_{\le i}$. It would then follow that $\ex{Z_{i}^\Pc\mid F_{i+1}, F_{\le i}}\approx  \ex{Y_1\mid Y_2, F_{\le i}}= \ex{Y_1\mid  F_{\le i}} \approx \ex{Z_{i}^\Pc\mid  F_{\le i}}$ as required. To this end, we use a second result by \citet{HNOSS18New}.\footnote{Assuming the nonexistence of key-agreement protocols, \cref{thm:intro:DecorolatorInf} implies \cref{thm:intro:ForcasterInf}. Yet, we chose to use both results to make the text more modular. }
\begin{theorem}[\citet{HNOSS18New}'s dichotomy, informal]\label{thm:intro:DecorolatorInf}
	Let $\Delta = \tpprot{\Ac}{\Bc}$ be an efficient single-bit output two-party protocol and assume infinitely-often key-agreement protocol does not exist. Then, for any constant $\rho>0$, there exists a poly-time algorithm (\textit{decorrelator}) $\Dcr$ mapping transcripts of $\Delta$ into $[0,1] \times [0,1]$ such that the following holds: let $(X,Y,T)$ be the parties' outputs and transcript in a random execution of $\Delta$, then
	$$(X,Y,T) \cindist_\rho (U_{p^{\Ac}},U_{p^{\Bc}},T)_{(p^{\Ac},p^{\Bc}) \la \Dcr(T)}.$$
\end{theorem}
Namely, assuming io-key-agreement does not exist, the distribution of the parties' output given the transcript is $\rho$-close to the product distribution given by $\Dcr$. We assume for simplicity that the theorem holds for \emph{many-bit} output protocols and not merely single bit (we get rid of this assumption in the actual proof).

We define another variant $\htPi$ of $\pi$ that internally uses the forecaster $\Fc$, and  show that the existence of a decorrelator for $\hPi$ implies that $F_{i+1}$ and $Z_{i}^\Pc$ are approximately independent given $F_{\le i}$, and \cref{eq:Dec:Ind} follows. For concreteness, we focus on party $\Pc=\Bc$. 

Fix $i$ such that $\Ac$ sends the $(i+1)$ message in $\pi$ and define protocol $\htPi= (\hAc,\hBc)$ according to the following specifications: the parties interact just as in $\pi$ for the first $i$ rounds; then $\hBc$ outputs the \ith backup value of $\Bc$ and $\hAc$ internally computes $m_{i+1}$ and outputs $f_{i+1}= \Fc(m_{\le i+1})$.   By \cref{thm:intro:DecorolatorInf} there exists  an efficient decorrelator $\hDcr$ for $\htPi$ \wrt $\rho = 1/r$. That is:

\begin{align}\label{eq:hDcrCloseToReal}
(F_{i+1},Z_i^{\Bc},M_{\le i}) \cindist_{1/r} (U_{p^{\hAc}},U_{p^{\hBc}},M_{\le i})_{(p^{\hAc},p^{\hBc}) \la \hDcr(M_{\le i})},
\end{align}
where now $p^{\hAc}$ describes a non-Boolean distribution, and $U_{p^{\hAc}}$ denotes an independent sample from this distribution. 

Since $\Fc$ and $\hDcr$ both output an estimate of (the expectation of) $Z_i^{\Bc}|M_{\le i}$ in a way that is indistinguishable from the real distribution of $Z_i^{\Bc}$ (given $M_{\le i}$), both algorithms output essentially the same value. Otherwise, the ``accurate'' algorithm can be used to distinguish the output of the ``inaccurate'' algorithm from the real output. It follows that 
\begin{align}\label{eq:hDcrCloseToInd}
(U_{p^{\hAc}},U_{p^{\hBc}},M_{\le i})_{(p^{\hAc},p^{\hBc}) \la \hDcr(M_{\le i})}\cindist_{1/r} (U_{p^{\hAc}},U_{F_i^{\Bc}},M_{\le i})_{(p^{\hAc},\cdot ) \la \hDcr(M_{\le i})}
\end{align} Using a data-processing argument in combination with \cref{eq:hDcrCloseToReal,eq:hDcrCloseToInd}, we deduce that
\begin{align}
\left(F_{i+1},Z_i^{\Bc},F_{\le i}\right) \cindist_{1/r} \left(U_{p^{\hAc}},U_{p^{\hBc}},F_{\le i}\right)_{(p^{\hAc},p^{\hBc}) \la \hDcr(M_{\le i})} \cindist_{1/r} \left(U_{p^{\hAc}},U_{F_i^{\Bc}},F_{\le i}\right)_{(p^{\hAc},\cdot) \la \hDcr(M_{\le i}) }
\end{align}  Finally, conditioned on $F_{\le i} $, we observe that the pair of random variables $(U_{p^{\hAc}},U_{F_i^{\Bc}})_{(p^{\hAc},\cdot) \la \hDcr(M_{\le i})}$ are mutually independent since $U_{F_i^{\Bc}}$ is sampled independently according to $F_i^{\Bc}$, and $F_i^{\Bc}$ is fully determined by $F_{\le i}$.

\subsection{Related Work}
We review some of the relevant work on fair coin flipping protocols.

\paragraph{Necessary hardness assumptions.}

This line of work examines the minimal assumptions required to achieve an $o(1/\sqrt{r})$-bias two-party coin flipping protocols, as done in this paper. The necessity of one-way functions for weaker variants of coin flipping protocol where the honest party is allowed to abort if the other party aborts or deviates from the prescribed protocol, were considered in~\cite{ImpagliazzoLu89,Maji10,HaitnerOmri14, BermanHT18}. 

More related to our bound, prior to our work, \citet{Dachman11} showed that any fully black-box construction of $O(1/r)$-bias two-party protocols based on one-way functions (with $r$-bit input and output) needs $\Omega(r/\log r)$ rounds, and \citet{DachmanMM14} showed that there is no fully black-box and function \textit{oblivious} construction of $O(1/r)$-bias two-party protocols from one-way functions (a protocol is function oblivious if the outcome of protocol is independent of the choice of the one-way function used in the protocol). For the case we are interested in, i.e.~constant-round coin flipping protocols, \cite{Dachman11} already yields that black-box techniques from one-way functions can only guarantee fairness of order $1/\sqrt{r}$.  Finally in a recent work, \citet{MW20} showed that in the random oracle model, any coin flipping protocol can be biased by $\Omega(1/\sqrt{r})$, implying that optimally fair coin flipping protocols cannot be constructed from one-way functions in a black-box manner.

\paragraph{Lower bounds.} 
\citet{Cleve86} proved that, for every $r$-round two-party coin flipping protocol, there exists an efficient adversary that can bias the output by $\Omega(1/r)$. \citet{CleveI93} proved that, for every $r$-round two-party coin flipping protocol, there exists an inefficient fail-stop adversary that biases the output by $\Omega(1/\sqrt{r})$. They also showed that a similar attack exists if the parties have access to an ideal commitment scheme. All above bounds extend to the multi-party case (with no honest majority) via a simple reduction. Very recently, \citet{BHMO18New} showed that \emph{any} $r$-round $\np$-parties coin flipping with $\np^k > r$, for some $k\in \N$, can be biased by $1/(\sqrt{r} \cdot (\log r)^k)$. Ignoring logarithmic factors, this means that if the number of parties is $r^{\Omega(1)}$, the majority protocol of \cite{AwerbuchBCGM1985} is optimal.

\paragraph{Upper bounds.} 
\citet{Blum83} presented a two-party two-round coin flipping protocol with bias $1/4$. \citet{AwerbuchBCGM1985} presented an $\np$-party $r$-round protocol with bias $O(\np/\sqrt{r})$ (the two-party case appears also in \citet{Cleve86}). \citet{MoranNS09} solved the two-party case by giving a two-party $r$-round coin flipping protocol with bias $O(1/r)$. \citet{HaitnerT14} solved the three-party case up to poly-logarithmic factor by giving a three-party coin flipping protocol with bias $O(\polylog(r)/r)$. \citet{BHLT17} showed an $\np$-party $r$-round coin flipping protocol with bias $\widetilde{O}(\np^3 2^\np/r^{\frac{1}{2}+\frac{1}{2^{\np-1}-2}})$. In particular, their protocol for four parties has bias $\widetilde{O}(1/r^{2/3})$, and for $\np = \log \log r$ their protocol has bias smaller than \citet{AwerbuchBCGM1985}.

For the case where less than $2/3$ of the parties are corrupt, \citet{BeimelOO15} showed
an $\np$-party $r$-round coin flipping protocol with bias $2^{2^k}/r$, tolerating up to $t=(\np+k)/2$ corrupt parties. \citet{AO16} showed an $\np$-party $r$-round coin flipping protocol with bias $\widetilde{O}(2^{2^\np}/r)$, tolerating up to $t$ corrupted parties, for constant $\np$ and $t<3\np/4$. 

\subsection{Open Questions}
We show that constant-round coin flipping protocol  with ``small'' bias (\ie $o(1/\sqrt r)$-fair, for $r$ round protocol) implies io-key-agreement. Whether such a reduction can be extended to protocols with super-constant round complexity remains open. The barrier to extending our results is that the dichotomy result of \citet{HNOSS18New} only guarantees indistinguishablility with constant advantage (as opposed to vanishing or negligible advantage). 

The question of reducing oblivious transfer to optimally-fair coin flip is also open. We recall that all known small  bias  coin flipping protocols rely on it \cite{MoranNS16,HaitnerT17, BHLT17}. It is open whether the techniques of \citet{HNOSS18New} can provide a similar dichotomy with respect to (io-) oblivious transfer (as opposed to io-key-agreement) allowing for the realization of oblivious transfer from  $o(1/\sqrt r)$-fair (constant round) coin flip via the techniques of the present paper. 

\Inote{see my edits}

\subsection*{Paper Organization}
Basic definitions and notation used through the paper, are given in \cref{sec:Prlim}. The formal statement and proof of the main theorem are given in \cref{sec:reduction}.

%% file: Prelims.tex
\section{Preliminaries}\label{sec:Prlim}

\subsection{Notation}\label{sec:prelim:notation}
We use calligraphic letters to denote sets, uppercase for random variables and functions,  lowercase for values. 
For $a,b\in \R$,  let $a\pm b$ stand for the interval $[a-b,a+b]$. For $n\in \N$, let $[n] = \set{1,\ldots,n}$ and $(n) = \set{0,\ldots,n}$. Let $\poly$ denote the set of all polynomials, let \ppt stand for probabilistic  polynomial time and   \pptm denote a \ppt algorithm (Turing machine).  A function $\nu \colon \N \to [0,1]$ is \textit{negligible}, denoted $\nu(n) = \negl(n)$, if $\nu(n)<1/p(n)$ for every $p\in\poly$ and large enough $n$. For a sequence $x_1,\ldots, x_r$ and $i\in [r]$, let $x_{\le i}=x_1,\ldots,x_i$ and   $x_{< i}=x_1,\ldots,x_{i-1}$. 



Given a distribution, or random variable,  $D$, we write $x\gets D$ to indicate that $x$ is selected according to $D$. Given a finite set $\cs$, let $s\la \cs$ denote that $s$ is selected according to the uniform distribution over $\cs$. The support of $D$, denoted $\Supp(D)$, be defined as $\set{u\in\Uni: D(u)>0}$. The \emph{statistical distance} between two distributions $P$ and $Q$ over a finite set $\Uni$, denoted as $\SD(P,Q)$, is defined as $\max_{\cs\subseteq \Uni} \size{P(\cs)-Q(\cs)} = \frac{1}{2} \sum_{u\in \Uni}\size{P(u)-Q(u)}$. Distribution ensembles $X=\set{X_\kappa}_{\kappa\in \N}$ and $Y=\set{Y_\kappa}_{\kappa\in \N}$ are \emph{$\delta$-computationally indistinguishable in the set $\I$}, denoted by $X\cindist_{\I,\delta} Y$, if for every \pptm $\Dc$ and sufficiently large $\kappa\in \I$: $\size{\pr{\D(1^\kappa, X_\kappa)=1} - \pr{\D(1^\kappa, Y_\kappa)=1} }\le \delta$.


\subsection{Protocols}\label{sec:prelim:protocols}
Let $\pi= (\Ac,\Bc)$  be a two-party protocol. The protocol $\pi$ is \ppt if the running time of both $\Ac$ and $\Bc$ is polynomial in their input length (regardless of the party they interact with). We denote by $(\Ac(x),\Bc(y))(z)$ a random execution of $\pi$ with private inputs $x$ and $y$, and common input $z$, and sometimes abuse notation and write $(\Ac(x),\Bc(y))(z)$ for the parties' output in this execution.

We will  focus on no-input two-party single-bit output \ppt protocol: the only input of the two \ppt parties is the common security parameter given in unary representation. At the end of the execution,  each party outputs a single bit.  Throughout, we assume \wlg that the transcript contains $1^\kappa$ as the first message. Let $\pi= (\Ac,\Bc)$  be such a two-party single-bit output protocol. 
For $\kappa\in \N$, let $C^{\Ac,\kappa}_\pi$, $C^{\Bc,\kappa}_\pi$ and $T^\kappa_\pi$ denote the outputs of $\Ac$, $\Bc$ and the transcript of $\pi$, respectively, in a random execution of $\pi(1^\kappa)$.  

\subsubsection{Fair Coin Flipping}
Since we are concerned with a lower bound, we only give the game-based definition of coin flipping protocols (see \cite{HaitnerT17} for the stronger simulation-based definition).

\begin{definition}[Fair coin flipping protocols]\label{def:FairCCF}
A  \ppt single-bit output  two-party protocol $\pi = (\Ac,\Bc)$ is an {$\eps$-fair coin flipping protocol}, if the following holds.
	\begin{description}
		\item[Output delivery:] The honest party always outputs a bit (even if the other party acts dishonestly, or aborts).

		\item[Agreement:] The parties always output the same bit in an honest execution. 		
		
		\item[Uniformity:]  $\pr{C^{\Ac,\kappa}_\pi = b} = 1/2$ (and thus  $\pr{C^{\Bc,\kappa}_\pi = b} = 1/2)$, for  both $b\in \zo$ and all $\kappa\in \N$.
\remove{	
	\footnote{The proof of our main result   easily extends to non optimal uniformity  condition. Say, if we  only require that $\pr{C^\Ac_\kappa = b}  \ge 1/4$ for both $b\in \zo$. }
		}
		
		\item[Fairness:] For any \ppt $\Ac^\ast$ and $b\in \zo$, for sufficiently large $\kappa \in \N$ it holds that
		
		$\pr{C^{\Bc,\kappa}_\pi =b}\le 1/2 + \eps$,  and the same holds for the output bit of $\Ac$.
	\end{description}
\end{definition}

\subsubsection{Key-Agreement}
We focus on single-bit output key-agreement protocols. 

\begin{definition}[Key-agreement protocols]\label{def:KA}
A  \ppt single-bit output  two-party protocol $\pi = (\Ac,\Bc)$  is  {\sf io-key-agreement}, if  there exist an infinite $\I\subseteq \N$, such that the following hold for $\kappa$'s in $\I$:
	\begin{description}
		
		\item[Agreement.]  $\pr{C^{\Ac,\kappa}_\pi =  C^{\Bc,\kappa}_\pi} \ge 1- \negl(\kappa)$.
		
		\item[Secrecy.] $\pr{\Ec(T^\kappa_\pi)=C^{\Ac,\kappa}_\pi} \le 1/2 + \negl(\kappa)$, for every \ppt \Ec.
		
	\end{description}

\end{definition}

\subsection{Martingales}

\begin{definition}[Martingales]\label{def:DMartingales} Let $X_0, \ldots, X_r$ be a sequence of random variables. 
We say that $X_0, \ldots, X_r$ is a martingale sequence if \/ $\ex{X_{i+1} \mid X_{\le i} = x_{\le i} } = x_{i} $ for every $i\in[r-1]$. 
\end{definition}

\noindent In plain terms, a sequence is a martingale if the expectation of the next point conditioned on the entire history is exactly the last observed point. One way to obtain a martingale sequence is by constructing a \textit{Doob martingale}. Such a sequence is defined by  $X_i = \ex{f(Z) \mid  Z_{\le i}}$, for arbitrary random variables $Z= (Z_1,\ldots,Z_r)$  and  a function $f$ of interest.  We will use the following fact proven by  \cite{CleveI93} (we use the variant as proven in \cite{BHMO18New}). 

\begin{theorem}\label{thm:MartingalesHaveJumpejumps}
Let $X_0,\ldots, X_r$ be a martingale sequence such that $X_i\in [0,1]$, for every $i\in [r]$. If $X_0=1/2$ and $\pr{X_r\in \set{0,1}}=1$, then $\pr{\exists i\in [r]\st  \size{X_i-X_{i-1}}\ge \frac{1}{4\sqrt{r}} }\ge \frac{1}{20}$.
\end{theorem}

%% file: ReductionToKA.tex
\newcommand{\rdm}{s}
\newcommand{\bu}{v}

\section{Fair Coin Flipping to Key-Agreement}\label{sec:reduction}
In this section, we prove our main result: if there exist constant-round coin flipping protocols which improve over the $1/\sqrt{r}$-bias majority protocol of \cite{AwerbuchBCGM1985}, then infinitely-often key-agreement exists as well. Formally, we prove the following theorem.
\def \mainthm{
The following holds for any (constant)  $r\in \N$: if there exists an $r$-round, $\frac{1}{25600\sqrt{r}}$-fair two-party coin flipping protocol, see \cref{def:FairCCF}, then there exists an  infinitely-often key-agreement protocol.}

\begin{theorem}\label{thm:main}
	\mainthm\footnote{\cref{def:FairCCF} requires perfect uniformity: the common output in an honest execution is an unbiased bit. The proof given below, however,  easily extends to any non-trivial uniformity condition, \eg the common output equals $1$ with  probability $3/4$.}\footnote{We remark that we did not optimize the value of the constant.}
\end{theorem}

Before formally proving \cref{thm:main}, we briefly recall the outline of the proof as presented in the introduction (we ignore certain constants in this outline). We begin with a good forecaster for the coin flipping protocol $\pi$ (which must exist, according to \cite{HNOSS18New}), and we define an efficiently computable conditional expected outcome sequence $X= (X_0,\ldots,X_r)$ for $\pi$, conditioned on the forecaster's outputs. Then, we show that (1) the \ith backup value (default output in case the opponent aborts) should be close to $X_i$; otherwise, an efficient attacker can use the forecaster  to bias the output of the other party  (this attack is applicable regardless of the existence of infinitely-often key-agreement). And (2), since $X$ is a martingale sequence, ``large'' $1/\sqrt{r}$-gaps are bound to occur in some round, with constant probability. Hence, combining (1) and (2), with constant probability, for some $i$, there is a $1/\sqrt{r}$-gap between $X_i$ and the forecasters' prediction for one party \emph{at the preceding round} $i-1$. Therefore, unless protocol $\pi$ implies io-key-agreement, the aforementioned gap can be exploited to bias that party's output by $1/\sqrt{r}$, by instructing the opponent to abort as soon as the gap is detected. In more detail, the success of the attack requires that (3) the event that a gap occurs is (almost) \emph{independent} of the backup value of the honest party.  It turns out that if $\pi$ does not imply io-key-agreement, this third property is guaranteed by the dichotomy theorem of \cite{HNOSS18New}. In summary, if io-key-agreement does not exist, then protocol $\pi$ is at best $1/\sqrt{r}$-fair.

\smallskip
Moving to the formal proof, fix an $r$-round, two-party coin flipping protocol $\pi= (\Ac,\Bc)$ (we assume nothing about its fairness parameter for now). We associate the following  random variables  with a random honest execution of  $\pi(1^\kappa)$.  Let $M^\pk = (M^\pk_1,\ldots,M^\pk_r)$ denote the messages of the protocol and let $C^\pk$ denote the (always) common output of the parties. For $i\in \set{0,\ldots,r}$ and $\Pc\in \set{\Ac,\Bc}$, let $Z_i^{\Pc,\pk}$ be the ``backup'' value party $\Pc$ outputs, if the other party aborts after the \ith message was sent. In particular, $Z_r^{\Ac,\pk} = Z_r^{\Bc,\pk} =C^\pk $ and    $\pr{C^\pk = 1}   = 1/2$. 


\paragraph{Forecaster for $\pi$.}

We are  using a \textit{forecaster} for $\pi$, guaranteed by the following theorem (proof  readily follows from  \citet[Thm 3.8]{HNOSS18New}).

\begin{theorem}[\citet{HNOSS18New}, existence of forecasters]\label{thm:Forcasters}
	Let  $\Delta$ be a no-input, single-bit output two-party protocol. Then for any constant $\rho>0$,  there exists a \ppt constant output-length algorithm  \Fc (forecaster) mapping  transcripts of  $\Delta$ into (the binary description of) pairs in $[0,1] \times [0,1]$ and an infinite set $\I \in \N$ such that the following holds:  let $C^{\Ac,\pk}$, $C^{\Bc,\pk}$ and $T^\pk$ denote the parties' outputs and protocol transcript, respectively, in a random execution of $\Delta(1^\pk)$. Let $m(\kappa)\in \poly$ be a bound on the number of coins used by $\Fc$ on  transcripts in $\supp(T^\pk)$, and let $S^\pk$ be a uniform string of length $m(\kappa)$. Then,
	
	\begin{itemize}
		\item 	$(C^{\Ac,\pk},T^\pk,S^\pk)  \cindist_{\rho,\I} (U_{p^{\Ac}},T^\pk,S^\pk)_{(p^{\Ac},\cdot) = \Fc(T^\pk;S^\pk)}$, and 
		
		\item 	$(C^{\Bc,\pk},T^\pk,S^\pk)  \cindist_{\rho,\I} (U_{p^{\Bc}},T^\pk,S^\pk)_{(\cdot,p^{\Bc}) = \Fc(T^\pk;S^\pk)}$. 
		
	\end{itemize}
letting $U_p$ be a Boolean random variable taking the value $1$ with probability $p$.\footnote{\citet{HNOSS18New} do not limit the output-length of $\Fc$. Nevertheless, by  applying \cite{HNOSS18New}  with  parameter $\rho/2$ and  chopping each of the forecaster's  outputs to the first  $\ceil{\log 1/\rho} +1$ (most significant) bits, yields the desired constant output-length  forecaster.} 
\end{theorem}

Since we require a forecaster for all (intermediate) backup values of $\pi$, we apply \cref{thm:Forcasters} \wrt  the following variant of protocol $\pi$, which simply stops the execution at a random round.

\begin{protocol}[$\tPi = \tpprot{\tAc}{\tBc}$]
\item Common input: security parameter $1^\pk$.

\item Description:

	\begin{enumerate}
		\item  $\tAc$ samples $i\la  [r]$ and sends it to $\tBc$.
		\item The parties interact in  the first $i$ rounds of a random execution of $\pi(1^\kappa)$, with  $\tAc$ and $\tBc$ taking the role of $\Ac$ and $\Bc$ receptively. 
		
		Let $z_i^\Ac$ and $z_i^\Bc$ be the \ith backup values of  $\Ac$ and $\Bc$ as computed by the parties in the above execution.

		\item $\tAc$ outputs  $z_i^\Ac$,  and $\tBc$ outputs $z_i^\Bc$.
	\end{enumerate}
\end{protocol}

 Let $\rho=10^{-6} \cdot r^{-5/2}$. Let  $\I\subseteq \N$ and a \ppt $\Fc$ be the infinite set and \ppt forecaster resulting by applying \cref{thm:Forcasters} \wrt protocol  $\tPi$ and  $\rho$, and let $S^\pk$ denote a long enough uniform string to be used by $\Fc$ on transcripts of $\tPi(1^\pk)$.  The following holds \wrt $\pi$.

\begin{claim}\label{claim:goodForcaster}
For  $I \la [r]$, it holds that
\begin{itemize}
		
 \item $( Z_I^{\Ac,\pk},M^{\pk}_{\le I},S^\pk)   \cindist_{\rho,\I} (U_{p^{\Ac}},M^\pk_{\le I},S^\pk)_{(p^{\Ac},\cdot) = \Fc(M_{\le I};S^\pk)}$, and

\item $ (Z_I^{\Bc,\pk},M^{\pk}_{\le I},S^\pk)   \cindist_{\rho,\I} (U_{p^{\Bc}},M^\pk_{\le I},S^\pk)_{(\cdot,p^{\Bc}) = \Fc(M_{\le I};S^\pk)}$,
\end{itemize}
letting $\Fc(m_{\le i};r) = \Fc(i,m_{\le i};r)$.
\end{claim}
\begin{proof}
	Immediate, by \cref{thm:Forcasters} and the definition of $\tpi$.
\end{proof}
We assume  \wlg that the common output appears on the last message of $\pi$ (otherwise, we can add a final message that contains this value, which does not hurt the security of $\pi$). Hence, \wlg it holds that  $\Fc(m_{\le r};\cdot) = (b,b)$, where $b$ is the output bit as implied by $m_{\le r}$ (otherwise, we can change $\Fc$ to do so without hurting its forecasting quality).

For $\pk \in\N$, we define the random variables  $F_0^\pk,\ldots, F_r^\pk$, by  
\begin{align}
F_i^\pk = (F_i^{\Ac,\pk},F_i^{\Bc,\pk}) =  \Fc(M_{\le i};S^\pk)
\end{align}

\paragraph{The expected outcome sequence.}  
To attack the protocol, it is useful to evaluate  at each round  the expected outcome of the  protocol conditioned on the forecasters' outputs so far. 
To alleviate notation, we assume that the value of $\pk$ is determined by $\size{S^\pk}$. 
\begin{definition}[the expected outcome function]\label{def:GameValueFunction}
For $\pk\in \N$, $i\in [r]$,  $f_{\le i} \in \supp(F^\pk_{\le i})$ and $\rdm\in \Supp(S^\pk)$, let   
	\begin{align*}
	g(f_{\le i},\rdm )=\ex{C^\pk\mid F^\pk_{\le i}=f_{\le i}, S^\pk =\rdm}.
	\end{align*}
\end{definition}
Namely, $g(f_{\le i},\rdm )$ is the probability that the output of the protocol in a random execution is $1$, given that  $\Fc(M_{\le j};s)= f_j$ for every  $j\in (i)$ and $M_1,\ldots,M_r$ being the transcript of this execution.

\paragraph{Expected outcome sequence is approximable.}
The following claim, proven in \cref{sec:computingGV}, yields that the expected outcome sequence  can be approximated efficiently.

\def\claimcomputingGV
{
	
	There exists  \pptm $\Gc$ such that
	
	$$\pr{\Gc(F^\pk_{\le i},S^\pk)\notin  g(F^\pk_{\le i} ,S^\pk) \pm \rho} \le \rho,$$
	 for every $\pk \in \N$ and  $i  \in [r]$.

}	

\begin{claim}[Expected outcome sequence is  approximable]\label{claim:computingGV}
	\claimcomputingGV
\end{claim}

Algorithm $\Gc$ approximates the value of $g$ on input $(f_{\le i},s)\in \supp(F^\pk_{\le i} ,S^\pk)$ by running multiple independent instances of protocol $\pi(1^\kappa)$ and keeping track of the number of times it encounters $f_{\le i}$ and the protocol outputs one. Standard approximation techniques yield that, unless $f_{\le i}$ is very unlikely, the output of $\Gc$ is close to $g(f_{\le i},s)$. \cref{claim:computingGV} follows by carefully choosing the number of iterations for $\Gc$ and bounding the probability of encountering an unlikely $f_{\le i}$.

\paragraph{Forecasted backup values  are close to  expected outcome sequence.}
The following claim bounds the probability that the expected outcome sequence and the forecaster's outputs deviate by more than $1/8\sqrt{r}$. The proof is given in \cref{sec:BUappxGV}.

\def\claimBUappxGV
{
Assuming $\pi$ is $\frac{1}{6400\sqrt{r}}$-fair,  then 
\begin{align*}
\pr{\exists  i \in [r] \st \size{g(F^\pk_{\le i} ,S^\pk) - F^{\Pc,\pk}_i}\ge 1/8\sqrt{r}}< 1/100
\end{align*}
 for both $\Pc\in \set{\Ac,\Bc}$ and  large enough  $\kappa\in \I$.
}	

\begin{claim}[Forecasted backup  values  are close to  expected outcome sequence]\label{claim:BUappxGV}
	\claimBUappxGV
\end{claim}

Loosely speaking, \cref{claim:BUappxGV} states that the expected output sequence and the forecaster's outputs are close for a fair protocol. If not, then either of the following attackers $\Pc^\ast_0$, $\Pc^\ast_1$ can bias the output of party $\Pc$: for fixed randomness $s\in \supp(S^\kappa)$, attacker $\Pc^\ast_z$ computes $f_{i}=\Fc(m_{\le i}, s)$ for partial transcript $m_{\le i}$ at round $i\in [r]$, and aborts as soon as $(-1)^{1-z}(\Gc(f^\pk_{\le i} , s) - f_i) \ge 1/8\sqrt{r}-\rho$. The desired bias is guaranteed by the accuracy of the forecaster (\cref{claim:goodForcaster}), the accuracy of algorithm $\Gc$ (\cref{claim:computingGV}) and the presumed frequency of occurrence of a suitable gap. The details of the proof are given in \cref{sec:BUappxGV}.

\paragraph{Expected outcome sequence  has large gap.}
Similarly to \cite{CleveI93}, the success of our attack depends on the occurrence of large gaps in the expected outcome sequence. The latter is guaranteed  by \cite{CleveI93} and \cite{BHMO18New}, since the expected outcome sequence is a suitable martingale.

\begin{claim}[Expected outcomes have large gap]\label{claim:GVjump}
	For every   $\kappa\in \N$, it holds that  $\pr{\exists i \in [r] \colon \size{g(F^\pk_{\le i} ,S^\pk) - g(F^\pk_{\le i-1} ,S^\pk)}\ge 1/4\sqrt{r}} >1/20$.
\end{claim}
\begin{proof}
	Consider the sequence of random variables $G_0^\pk,\ldots,G_r^\pk$ defined by 
	$G_i^\pk  = g(F^\pk_{\le i}, S^\pk)$. Observe that this is a Doob (and hence, strong) martingale sequence, \wrt the random variables $Z_0 = S^\pk$ and $Z_i = F^\pk_{i}$ for $i\in[r]$, and the function $f(S^\pk, F^\pk_{\le r})= g(F^\pk_{\le r}, S^\pk) = F^\pk_r[0]$ (\ie the function that outputs the actual output of the protocol, as implied by $F^\pk_r$). 
	Clearly,  $G_0^\pk = 1/2$ and $G_r^\pk\in \zo$ (recall that we assume that $\Fc(M_{\le r};\cdot) =  (b,b)$, where $b$ is the output bit as implied by $M_{\le r}$). Thus, the proof follows by \cref{thm:MartingalesHaveJumpejumps}.
\end{proof}

\paragraph{Independence of attack decision.}
\cref{claim:goodForcaster} immediately yields that the expected values of $F_{ i }$ and  $Z_{i}^\Pc$ are close, for both $\Pc \in \set{\Ac,\Bc}$ and every $i\in[r]$. Assuming io-key-agreement  does not exist, the following claim essentially states that $F_{ i }$ and  $Z_{i}^\Pc$ remain close in expectation, even if we condition on  some event that depends on the other party's next message.  This observation will allow us to show that, when a large gap in the expected outcome is observed by one of the parties, the (expected value of the) backup value of the other party still lags behind. The following claim captures the core of the novel idea in our attack, and its proof is the most technical aspect towards proving our main result. 

\def\claimAttackOpurunity
{
	Let  $\Cc$ be a single-bit  output \pptm. For $\pk \in \N$ and  $\Pc \in \set{\Ac,\Bc}$, let $E_1^{\Pc,\pk},\ldots,E_{r}^{\Pc,\pk}$ be the sequence of random variables defined by $E_i^{\Pc,\pk} =  \Cc(F^\pk_{\le i}, S^\pk)$  if $\Pc$ sends the \ith message in $\pi(1^\pk)$, and $E_i^{\Pc,\pk} =0$ otherwise.
	
	Assume io-key-agreement protocols do not exist. Then, for any $\Pc \in \set{\Ac,\Bc}$ and infinite subset $\I' \subseteq \I$, there exists an infinite set $\I'' \subseteq \I'$ such that 
	\begin{align*}
	\ex{  E_{i+1}^{\Pc,\pk} \cdot (Z_{i}^{\Po,\pk} -  F^{\Po,\pk}_{  i}) } \in \pm 4r\rho
	\end{align*} 
	for every   $\kappa\in \I''$ and  $i\in (r-1)$, where $\Po$ denotes (the party in) $\set{\Ac,\Bc} \setminus\set{\Pc}$.
}

\begin{claim}[Independence of attack decision]\label{claim:AttackOpurunity}
	\claimAttackOpurunity
\end{claim}
Since  $\ex{  E_{i+1}^{\Pc,\pk} \cdot( Z_{i}^{\Po,\pk} -  F^{\Po,\pk}_{  i} )} = \ex{E_{i+1}^{\Pc,\pk} \cdot \ex{Z_{i}^{\Po,\pk} - F_{i}^{\Pc,\pk} \mid E_{i+1}^{\Pc,\pk} =1}}$,   \cref{claim:AttackOpurunity}  yields that the expected values of $F_{ i }$ and  $Z_{i}^\Pc$ remain close, even when   conditioning  on a likely enough event over the next message of $\Pc$.

The proof of \cref{claim:AttackOpurunity} is given in \cref{sec:AttackOpurunity}. In essence, we use the recent dichotomy of \citet{HNOSS18New} to assert that if io-key-agreement does not exist, then the values of  $E_{i+1}^{\Pc,\pk}$ and  $Z^{\Po,\pk}_{  i}$ conditioned on $M_{\le i}$ (which determines the value of $F^{\Po,\pk}_{ i}$), are (computationally) close to be in a product distribution.

\paragraph{Putting everything together.}
Equipped with   the above observations, we  prove  \cref{thm:main}.
\begin{proof}[Proof of \cref{thm:main}]
	
	Let $\pi$ be an $\eps =  \frac{1}{25600\sqrt{r}}$-fair coin flipping protocol. By \cref{claim:BUappxGV,claim:GVjump},  we can  assume \wlg that there exists an infinite subset $\I' \subseteq \I$ such that  
	\begin{align}\label{final:0}
	\pr{\exists i \in [r] \colon      \mbox{ $\Ac$ sends  \ith message in $\pi(1^\pk)$} \land g(F^\pk_{\le i} ,S^\pk) - F^{\Bc,\pk}_{i-1} \ge \frac{1}{8\sqrt{r}}} \ge \frac{1}{80} -\frac{1}{100} = \frac{1}{400}
	\end{align}
	We define the  following   \ppt fail-stop attacker $\As$ taking the role of  $\Ac$ in $\pi$. We will show below that assuming io-key-agreement do not exist, algorithm $\As$ succeeds in biasing the output of $\Bc$ towards zero by $\eps$  for all $\pk \in \I''$, contradicting the presumed fairness of $\pi$. 
	
	 In the following, let $\Gc$ be the \pptm guaranteed to exist by \cref{claim:computingGV}.
	
	\begin{algorithm}[\As]\label{algo:FinallAttack} 
		\item Input: security parameter $1^\kappa$.
		
		\item Description: 
		\begin{enumerate}
			
			\item Sample $\rdm\la S^\pk$ and start a random execution of $\Ac(1^\pk)$.
			
			\item  Upon  receiving the  $(i-1)$ message $m_{i-1}$, do
			
			\begin{enumerate}
				\item Forward $m_{i-1}$ to $\Ac$, and let $m_i$ be the next message  sent by $\Ac$. 
				
				\item Compute $f_{i}= (f_{i}^\Ac,f_{i}^\Bc) =\Fc(m_{\le i}, \rdm)$.

				\item Compute  $\wt{g}_i = \Gc( f_{\le i}, \rdm)$.

				\item If $\wt{g}_i\ge f_{i-1}^\Bc + 1/16\sqrt{r}$,   abort (without sending further messages).

				Otherwise, send $m_i$ to $\Bc$ and proceed to the next round. 
				
			\end{enumerate} 
			
		\end{enumerate}
	\end{algorithm}  
	
	It is clear that $\As$ is a \pptm.  We conclude the proof  showing that  assuming io-key-agreement do not exist, $\Bc$'s output when interacting with  $\As$ is  biased towards zero by at least  $\eps$.

	The following random variables are defined \wrt a a random execution of  $(\As,\Bc)(1^\pk)$. Let $S^\pk$ and $F^\pk = (F_1^\pk,\ldots,F^\pk_r)$ denote the values of $s$ and $f_1,\ldots,f_r$ sampled by $\As$. Let $Z^{\Bc,\pk} = (Z_1^{\Bc,\pk},\ldots,Z^{\Bc,\pk}_r)$ denote the backup values computed by  $\Bc$. For $i\in [r]$, let $E_i^\pk$ be the event  that  $\As$ decides to abort in round $i$.  Finally, let  $J^\pk $ be the index $i$ with $E_i^\pk =1$, setting it to $r+1$ if no such index exist. Below, if we do not quantify over $\kappa$, it means that the statement holds for any $\kappa\in \N$.
	
	By \cref{claim:computingGV,final:0},
	\begin{align}\label{eq:final1}
	\pr{J^\pk \neq r+1}>  \frac{1}{400} - \rho \ge \frac{1}{800} 
	\end{align}
 for every $\pk \in \I'$. Where since the events $E_i^\pk$ and $E_j^\pk$  for $i\neq j$ are disjoint, 
	\begin{align}\label{eq:final3}  
	\ex{Z_{J^\kappa-1}^{\Bc,\pk} - F_{J^\kappa-1}^{\Bc,\pk}}&=  \ex{\sum_{i=1}^{r+1} E_i^\pk \cdot (Z_{i-1}^{\Bc,\pk} -  F_{i-1}^{\Bc,\pk} )}\\
	& = \sum_{i=1}^{r+1} \ex{E_i^\pk \cdot (Z_{i-1}^{\Bc,\pk} -  F_{i-1}^{\Bc,\pk} )}\nonumber\\
	&= \sum_{i=1}^{r} \ex{E_i^\pk \cdot (Z_{i-1}^{\Bc,\pk} -  F_{i-1}^{\Bc,\pk} )}.\nonumber
	\end{align}
	The last inequality holds since the protocol's output appears in the last message, by  assumption, and thus \wlg $Z_{r}^{\Bc,\pk} = F_{r}^{\Bc,\pk}$.  Consider the  single-bit  output \pptm $\Cc$  defined as follows: on input $(f_{\le i} =((f_1^\Ac,f_1^\Bc), \ldots,(f_i^\Ac,f_i^\Bc)),s)$, it outputs $1$ if $ \Gc(f_{\le i}, s) -  f^{\Bc}_{ i-1}  \ge  1/16\sqrt{r}$, and $ \Gc(f_{\le j}, s) -  f^{\Bc}_{ j-1}  <  1/16\sqrt{r}$ for all $j<i $; otherwise, it outputs  zero. Observe that  $E_i^\pk$ is the indicator of the event $\Ac$ sends the \ith message in $\pi(1^\kappa)$ and  $\Cc(F^\pk_{\le i}, S^\pk)=1$, for any fixing of  $(F^\pk,S^\pk, Z^{\Bc,\pk})$. Thus, assuming io-key-agreement protocols do not exist,   \cref{claim:AttackOpurunity} yields that that there exists an infinite set $\I'' \subset \I'$ such that 
	\begin{align}\label{eq:final4} 
	\ex{E_{i+1}^\pk\cdot (Z_i^{\Bc,\pk} - F_{i}^{\Bc,\pk})} \in \pm 4r\rho
	\end{align}
	for every $\kappa \in \I''$ and  $i\in [r-1]$. Putting together  \cref{eq:final3,eq:final4}, we conclude that 
	\begin{align}\label{eq:final5}  
	\ex{Z_{J^\kappa-1}^{\Bc,\pk} - F_{J^\kappa-1}^{\Bc,\pk}} \in \pm 4r^2\rho
	\end{align}
	for every $\kappa \in \I''$.

	Recall that our goal is to show that  $\ex{Z_{J^\kappa-1}^{\Bc,\pk}} $ is significantly smaller than $1/2$.  We  do it by showing that it is significantly smaller than $\ex{g(F^\pk_{\le J^\pk}, S^\pk)}$ which equals $1/2$, since, by tower law (total expectation),
	
		\begin{align}\label{eq:final6}
		\ex{g(F^\pk_{\le J^\pk}, S^\pk)}&=   \ex{C^\pk} =1/2 .
	\end{align}
	 Finally, let $G_i$ be the value of $ \Gc(F_{\le i},S^\pk)$ computed by $\As$ in the execution of  $(\As,\Bc)(1^\pk)$ considered above, letting $G_{r+1}= g(F^\pk_{\le r+1},S^\pk)$.
	 \cref{claim:computingGV} yields that
	\begin{align}\label{eq:final7}
	\ex{g(F^\pk_{\le J^\pk}, S^\pk)-G_{J^\pk}}\le 2r\rho
	\end{align}
	Putting all the above observations together, we conclude that, for every $\kappa \in \I''$,
	\begin{align*}
	& \lefteqn{\ex{Z_{J^\kappa-1}^{\Bc,\pk}} }\\ 
	&= \ex{g(F^\pk_{\le J^\pk}, S^\pk)} - \ex{G_{J^\pk}- F_{J^\kappa-1}^{\Bc,\pk}} + \ex{Z_{J^\kappa-1}^{\Bc,\pk} - F_{J^\kappa-1}^{\Bc,\pk}} -\ex{g(F^\pk_{\le J^\pk}, S^\pk)-G_{J^\pk}}  \\
	& \le \frac{1}{2} - \ex{G_{J^\pk}-F_{J^\kappa-1}^{\Bc,\pk}\mid J^\kappa \neq r+1}\cdot \pr{J^\pk\neq r+1} + 4r^2\rho +2r\rho \\
	& \le \frac{1}{2} - (1/16\sqrt{r}) \cdot (1/800) + 4r^2\rho +2r\rho \\
	& < \frac{1}{2} - \frac{1}{25600\sqrt{r}}.
	\end{align*}
	The first inequality holds by \cref{eq:final6,eq:final5,eq:final7}.
	The second inequality holds  by the definition of $J^\pk$ and   
	\cref{eq:final1}. The last inequality holds by our choice of $\rho$.

\end{proof}


\subsection{Approximating the Expected Outcome Sequence}\label{sec:computingGV}

In this section we prove \cref{claim:computingGV}, restated below.
\begin{claim}[\cref{claim:computingGV}, restated]\label{claim:computingGVRes}
	\claimcomputingGV
\end{claim}

The proof of \cref{claim:computingGVRes} is straightforward. Since  there are only constant number of rounds and  $\Fc$ has constant output-length, when fixing the randomness of \Fc,  the domain of $\Gc$ has constant size. Hence, the value of of $g$ can be approximated well via sampling. Details below.


Let $c$ be a bound on the number of possible outputs of $\Fc$ (recall that $\Fc$ has constant output-length). We are using the following implementation for $\Gc$.

\newcommand{\Fo}{\overline{\Fc}}
In the following, let $\Fo((m_1,\ldots,m_i);s) = (\Fc(m_1;s),\ldots , (\Fc(m_i;s))$ (\ie $\Fo(M_{\le i};S^\pk) = F_{\le i}$).
\begin{algorithm}[$\Gc$] \label{algo:allgamma} 
	\item Parameters: $v = \ceil{\frac{1}{2}\cdot \left (\frac{2c^{r}}{\rho}\right )^4 \cdot \ln\left (\frac{8}{\rho}\right )}$.  
	\item Input:  $f_{\le i}\in\supp(F_{\le i}^{\pk})$ and $s\in \Supp(S^\pk)$.
	\item Description:
	\begin{enumerate}
		\item Sample $v$  transcripts $\set{m^j,c^j}_{ j \in [v]}$  by taking the (full) transcripts and outputs of  $v$ independent executions  of  $\pi(1^{\pk})$. 
		
		\item For every $j\in [v]$ let  $f^j_i=\Fo(m_{\le i}^j;s)$.
		
		\item Let  $q=\size{\set{j \in [v] \colon f_{\le i}^j=f_{\le i}}}$ and $p=\size{\set{j \in [v] \colon f_{\le i}^j=f_{\le i}  \land c^j=1}}$.
		
		\item Set $\wt{g}=p/q$. (Set  $\wt{g}=0$ if $q=p = 0$.)
		
		\item Output $\wt{g}$.
	\end{enumerate} 
\end{algorithm}

\begin{remark}[A more efficient approximator.] 
	The running time of algorithm  $\Gc$ above is exponential in $r$. While this does not pose a problem for our purposes here, since $r$ is constant, it might leave the impression that out approach cannot be extended to protocols with super-constant round complexity.  So it is worth mentioning that the running time of \Gc can be reduced to be polynomial in $r$, by using the augmented weak martingale paradigm of \citet{BHMO18New}. Unfortunately, we currently cannot benefit from this improvement, since the result of \cite{HNOSS18New} only guarantees  indistinguishablity for constant $\rho$, which makes it useful only for attacking constant-round  protocols.
\end{remark}

We prove \cref{claim:computingGVRes} by showing that the above algorithm approximates $g$ well.

\begin{proof}[Proof of \cref{claim:computingGVRes} ] 
	
	To prove the quality of $\Gc$ in approximating  $g$, it suffices to prove the claim for every every $\kappa\in\N$, $i\in [r]$ and  fixed $\rdm\in \supp(S^\pk)$. That is 
	\begin{align}\label{eqn:gooal}
	\pr{\size{g(\Fo(M_{\le i},\rdm), \rdm) - \Gc(\Fo( M_{\le i},\rdm), \rdm) }\ge \rho } \le {\rho},
	\end{align}
	 where the probability is also taken over the random coins of $\Gc$.
	

	Fix $\kappa \in \N$ and omit it from the notation, and fix $i\in [r]$ and $s\in S^\pk$. Let $\cD_i  = \set{f_{\le i} \colon \pr{\Fo(M_{\le i},\rdm)=f_{\le i}}\ge \rho/2c^r}$.   By Hoeffding's inequality \cite{Hoeffding63}, for every $f_{\le i} \in \cD$,  it holds that

	\begin{align}
	\pr{\size{ g(f_{\le i}, \rdm)- \Gc( f_{\le i}, \rdm)} \ge \rho }&\le 4\cdot \exp\left (-2\cdot v\cdot \left (\rho/2c^r\right )^4\right )\\
	& \le  4\cdot \exp\left (-\frac{v \rho^4}{8c^{4r}}\right )\nonumber\\
	&\le  \rho/2.\nonumber
	\end{align}
	  \Nnote{There are a few missing steps above: We are applying Hoefding to the numerator and denominator of a fraction. Namely, $\pr{ \size{ \ldots - q/v} \ge (\rho/2c^r)^2 }  \le 2e^{-2v(\rho/2c^r)^4}$ and $\pr{ \size{ \ldots - p/v}\ge (\rho/2c^r)^2 }\le  2e^{-2v(\rho/2c^r)^4}$. By union bound and the approximation $1/(1+x)\in 1\pm 2x$, it follows that $4e^{-2v(\rho/2c^r)^4}\ge \pr{ \size{ g(f_{\le i}, \rdm) - p/q}\ge 3\cdot (\rho/2c^r)}\ge \pr{ \size{ g(f_{\le i}, \rdm) - p/q}\ge \rho}$ -- QED}
	It follows that
	\begin{align*}
	\lefteqn{\pr{\size{ g(\Fo(M_{\le i},\rdm), \rdm)- \Gc(\Fo(M_{\le i},\rdm), \rdm)} \ge \rho }}\\
	&\le \pr{(\Fo(M_{\le j},\rdm)\notin \cD} +  \rho/2\\
	&\le  \size{\Supp(\Fo(M_{\le j},\rdm))} \cdot\rho/2c^r   +  \rho/2\\
	&\le  c^r \cdot \rho /2c^r   +    \rho/2=  \rho.
	\end{align*}
\end{proof}

\subsection{Forecasted Backup Values are Close to Expected Outcome Sequence}\label{sec:BUappxGV}
In this section, we prove \cref{claim:BUappxGV} (restated below).
\begin{claim}[\cref{claim:BUappxGV}, restated]\label{claim:BUappxGVRes}
	\claimBUappxGV
\end{claim}

\begin{proof}
	
	Assume the claim does not holds for $\Pc = \Bc$ and infinitely many security parameters $\I$ (the case $\Pc=\Ac$ is proven analogously). That is, for all $\kappa \in \I$ and \wlg, it holds that 
	\begin{align}\label{eqn:badbakop}
	\pr{\exists  i \in [r] \st  g(F^\pk_{\le i} ,S^\pk) - F_i^{\Bc, \pk} \ge \frac{1}{8\sqrt{r}}}\ge  \frac{1}{200}
	\end{align}
	
	Consider  the following   \ppt fail-stop attacker $\As$ taking the role of  $\Ac$ in $\pi$ to bias the output of $\Bc$ towards zeros.
	
	\begin{algorithm}[\As]\label{algo:simpatak}  
		\item Input: security parameter $1^\kappa$.
		
		\item Description: 
		\begin{enumerate}
			
			\item Samples $\rdm\la S^\pk$ and start a random execution of $\Ac(1^\pk)$.
			
			\item For $i=1\ldots r$:
			
			After sending (or receiving) the prescribed message $m_i$:
			\begin{enumerate}
				\item Let $f_i=\Fc(m_{\le i}; \rdm)$ and $\mu_i=\Gc(f_{\le i}, \rdm)- f_i$.
				
				\item Abort if $\mu_i\ge \frac{1}{8\sqrt{r}} - \rho$  (without sending further messages).

				Otherwise, proceed to the next round. 
				
			\end{enumerate} 
			
		\end{enumerate}
	\end{algorithm}  
	
	In the following, we fix a large enough $\kappa\in \I$ such that \cref{eqn:badbakop} holds, and we omit it from the notation when the context is clear. We show that algorithm $\As$ biases the output of $\Bc$  towards zero by at least $1/(6400\sqrt{r})$.

	We associate the  following random variables with a random execution of $(\As,\Bc)$. Let $J$ denote the index where the adversary aborted,  \ie the smallest $j$ such that $\Gc(F_{\le j}, S)-F^{\Bc}_{j}\ge \frac{1}{8\sqrt{r}}-\rho$, or ${J}=r$ if no abort occurred. The following expectations are taken over $(F_{\le i}, S)$ and the random coins of $\Gc$. We bound $\ex{Z^{\Bc}_{J}}$, i.e.~the expected output of the honest party. 
	\begin{align}\label{eqn:basis}
	\lefteqn{\ex{Z^{\Bc}_{J}}}\\
	&   = \ex{Z^{\Bc}_{J} } + \ex{g(F_{\le {J}}, S) } - \ex{g(F_{\le {J}}, S) } + \ex{\Gc(F_{\le {J}}, S) -F^{\Bc}_{{J}} } - \ex{\Gc(F_{\le {J}}, S) - F^{\Bc}_{{J}} } \nonumber\\
	&  = \ex{g(F_{\le {J}}, S) }- \ex{\Gc(F_{\le {J}}, S) - F^{\Bc}_{{J}} } + \ex{\Gc(F_{\le {J}}, S) - g(F_{\le {J}}, S) } + \ex{Z^{\Bc}_{J} - F^{\Bc}_{{J}} }  \nonumber\\
	&  = \frac{1}{2}-\ex{\Gc(F_{\le {J}}, S) - F^{\Bc}_{{J}} }  + \ex{\Gc(F_{\le {J}}, S) - g(F_{\le {J}}, S) } + \ex{Z^{\Bc}_{J} - F^{\Bc}_{{J}} }.\nonumber
	\end{align}
	The last equation follows from $\ex{g(F_{\le {J}}, S)}=\ex{C}$ and thus $\ex{g(F_{\le {J}}, S) }=\frac{1}{2}$ (for a more detailed argument see \cref{eq:final6} and preceding text). We  bound each of the terms above separately. First, observe that 
	\begin{align}
	\lefteqn{\pr{{J}\neq r}}\\
	&\ge \pr{\left (\forall i\in [r]\colon \size{\Gc(F_{\le i}, S) -g(F_{\le i}, S)}\le\rho  \right ) \land\left (\exists j\in [r]\colon g(F_{\le j}, S)- F_j^{\Bc} \ge \frac{1}{8\sqrt{r}}\right )} \nonumber\\
	& \ge \pr{\exists j\in [r] \colon g(F_{\le j}, S)- F_j \ge \frac{1}{8\sqrt{r}} } - \pr{\exists i\in [r]\colon \size{\Gc(F_{\le i}, S) -g(F_{\le i}, S)}>\rho  } \nonumber\\
	& \ge \frac{1}{200} - \rho \nonumber\\
	&\ge \frac{1}{400} .\nonumber
	\end{align}
	The penultimate  inequality is by  \cref{eqn:basis,claim:computingGV}. It follows that  
	\begin{align}\label{eqn:fingap}
	\ex{g(F_{\le {J}}, S) - F^{\Bc}_{J} }&= \pr{{J}\neq r}\cdot \ex{g(F_{\le {J}}, S) - F^{\Bc}_{J}  \mid J \ne r}  \\
	&\ge    \frac{1}{400} \cdot  \left (\frac{1}{8\sqrt{r}} - \rho\right ) - \ex{\Gc(F_{\le {J}}, S) -g(F_{\le {J}}, S) }\nonumber\\
	&\ge  \frac{1}{400} \cdot  \frac{1}{8\sqrt{r}}  - 3\rho. \nonumber
	\end{align}
	The penultimate  inequality is by  \cref{claim:computingGV}. Finally,  since we were taking $\kappa$ large enough, \cref{claim:goodForcaster} and a data-processing argument yields that
	\begin{align}\label{eqn:err2}
	\ex{Z^{\Bc}_{J} -F^{\Bc}_{J} }\le r\rho
	\end{align}
	\Nnote{FULL DETAILS of \cref{eqn:err2}: if not true, take a distinguisher $\di$ that checks whether $I=J$ and outputs the second bit. More formally, on input $(\cdot, z, (m_{\le i},s)) $, algorithm $\di$ operates as follows: (1) it computes $f_{j}=\Fc(m_{\le j},s)$ and $\wt{g}_j=\Gc(f_\le j, s)$, for every $j\le i$ and (2)  outputs $z$ if $\wt{g}_i-f_i\ge 1/8\sqrt{r}-\rho$ and $\wt{g}_j-f_j< 1/8\sqrt{r}-\rho$ for all $j<i$ (3) outputs $\perp$ otherwise. If $\ex{Z^{\Bc}_{J} -F^{\Bc}_{J} }> r\rho$, then $\pr{\di((Z_I^{\Ac,\pk} ,Z_I^{\Bc,\pk}  , M^\pk_{\le I},S^\pk))=1}-\pr{\di((U_{F_I^{\Ac,\pk}},U_{F_I^{\Bc,\pk}} , M^\pk_{\le I},S^\pk))=1} =\ex{Z^{\Bc}_{J}  -F^{\Bc}_{J}  }\cdot \pr{J=I}>\rho$  }   
	
	We conclude that $\ex{g(F_{\le {J}}, S) - F^{\Bc}_{J} } \ge  \frac{1}{400} \cdot  \frac{1}{8\sqrt{r}}  - (r+ 3)\rho> 1/(6400\sqrt{r})$, in contradiction to the assumed fairness of $\pi$.

\end{proof}

\newcommand{\Sh}{\widehat{S}}

\subsection{Independence of Attack Decision}\label{sec:AttackOpurunity}
In this section, we prove \cref{claim:AttackOpurunity} (restated below).
\begin{claim}[\cref{claim:AttackOpurunity}, restated]\label{claim:AttackOpurunityRes}
	\claimAttackOpurunity
\end{claim} 
We prove for $\Pc = \Ac$. Consider the following  variant of $\pi$ in which the party playing  $\Ac$ is outputting $E^\Ac_i$ and the party playing $\Bc$ is outputting its backup value.

\begin{protocol}[$\hPi = \tpprot{\hAc}{\hBc}$]\label{prot:hpi} 
	\item Common input: security parameter $1^\pk$.
	
	\item Description:
	\begin{enumerate}
		
		\item Party $\hAc$ samples $i\la [r]$ and $\rdm\la S^\pk$, and sends them to  $\hBc$.

		\item  The parties interact in  the first $i-1$ rounds of a random execution of $\pi(1^\kappa)$, with $\hAc$ and $\hBc$ taking the role of $\Ac$  and $\Bc$ respectively.

		Let $m_1,\ldots,m_{i-1}$ be the messages, and let $z_{i-1}^\Bc$ be   the $(i-1)$ backup output of $\Bc$ in the above execution.

		\item  $\hAc$  sets  the value of $e^\Ac_i$ as follows:

		If $\Ac$ sends the $i-1$ message above,  then it sets $e^\Ac_i =0$.
		
		Otherwise, it
		\begin{enumerate}
			\item Continues the above execution of $\pi$ to compute  its next message $m_i$. 
			
			\item Computes $f_{i} =\Fc(m_{\le i}, \rdm)$.

			\item Let $e^\Ac_i  = \Cc(f_{\le i},\rdm)$.
		\end{enumerate}  

		\item  $\hAc$ outputs $e^\Ac_i$ and  $\Bc$ outputs $z_{i-1}^\Bc$.

	\end{enumerate}

\end{protocol}

We apply the  the following dichotomy  result of   \citet{HNOSS18New} on the above protocol.
\begin{theorem}[\citet{HNOSS18New}, Thm. 3.18,  dichotomy of two-party protocols]\label{thm:Decor}
	Let  $\Delta$ be an efficient single-bit output   two-party protocol. Assume io-key-agreement protocol do not exist,  then for any constant $\rho>0$ and infinite subset $\I \subseteq \N$, there exists a \ppt algorithm  \Dcr (decorelator) mapping  transcripts of  $\Delta$ into (the binary description of) pairs in $[0,1] \times [0,1]$ and an infinite set $\I' \in \N$, such that the following holds:  let $C^{\Ac,\pk}$, $C^{\Bc,\pk}$ and $T^\pk$ denote the parties' output and protocol transcript in a random execution of $\Delta(1^\pk)$. Let $m(\kappa)\in \poly$ be a bound on the number of coins used by $\Dcr$ on  transcripts in $\supp(T^\pk)$, and let   $S^\pk$ be a uniform string of length $m(\kappa)$. Then
	
	$$(C^{\Ac,\pk},C^{\Bc,\pk},T^\pk,S^\pk)  \cindist_{\rho,\I'} (U_{p^{\Ac}},U_{p^{\Ac}},T^\pk,S^\pk)_{(p^{\Ac},p^{\Bc}) = \Dcr(T^\pk;S^\pk)}$$
	letting $U_p$ be a Boolean random variable taking the value $1$ with probability $p$.
\end{theorem}

\begin{proof}[Proof of \cref{claim:AttackOpurunityRes}]   
	Assume io-key-agreement does not exits, and let  $\I''\subseteq \I'$ and a \ppt $\Dcr$ be the infinite set and \ppt decorrelator resulting by applying \cref{thm:Decor} \wrt protocol  $\hPi$ and  $\rho$. Let $\Sh^\pk$ denote a long enough uniform string to be used by $\Dcr$ on transcripts of $\hPi(1^\pk)$. Then for $I \la  (r-1)$, it holds that 
	\begin{align}\label{eq:Decor}
	( E_{I+1}^{\Ac,\pk},Z_{I}^{\Bc,\pk},M^{\pk}_{\le i},S^\pk,\Sh^\kappa)   \cindist_{\rho,\I''}  (U_{p^{\Ac}},U_{p^{\Bc}},M^\pk_{\le I}, S^\pk,\Sh^\kappa)_{(p^{\Ac},p^{\Bc}) = \Dcr(M_{\le I},S^\pk;\Sh^\pk)}
	\end{align}
	letting $\Dcr(m_{\le i},s;\widehat{s}) = \Dcr(i,s,m_{\le i};\widehat{s})$.

	For $i\in [r]$, let  $W_i^\pk = (W_i^{\Ac,\pk},W_i^{\Bc,\pk}) = \Dcr(M_{\le i},S^\pk;\Sh^\pk)$. The proof of  \cref{claim:AttackOpurunityRes1}  follows by the following three observations, proven below,  that hold for  large enough  $\pk \in \I ''$.  
	
	\begin{claim}\label{claim:AttackOpurunityRes1}
		$\ex{E_{I+1}^{\Ac,\pk}\cdot Z_{I}^{\Bc,\pk} - W_{I}^{\Ac,\pk}\cdot W_{I}^{\Bc,\pk}}\in \pm \rho$.
	\end{claim}

	\begin{claim} \label{claim:AttackOpurunityRes3}
		$\ex{W_{I}^{\Ac,\pk}\cdot F_{I}^{\Bc,\pk} - E_{I+1}^{\Ac,\pk}\cdot F_{I}^{\Bc,\pk}}\in \pm \rho$.
	\end{claim}
	
	\begin{claim} \label{claim:AttackOpurunityRes2}
		$\ex{W_{I}^{\Ac,\pk}\cdot W_{I}^{\Bc,\pk} - W_{I}^{\Ac,\pk}\cdot F_{I}^{\Bc,\pk}}\in \pm 2\rho$.
	\end{claim}
We conclude that $\ex{  E_{I+1}^{\Pc,\pk} \cdot Z_{I}^{\Po,\pk} - E_{I+1}^{\Pc,\pk}\cdot F^{\Po,\pk}_{I} } \in \pm 4\rho$, and thus $\ex{  E_{i+1}^{\Pc,\pk} \cdot Z_{i}^{\Po,\pk} - E_{i+1}^{\Pc,\pk}\cdot F^{\Po,\pk}_{  i} } \in \pm 4r\rho$ for every $i\in (r-1)$. 
\end{proof}

\paragraph{Proving \cref{claim:AttackOpurunityRes1}.}
\begin{proof}[Proof of \cref{claim:AttackOpurunityRes1}]   
	Consider algorithm $\Dc$ that on input $(z^\Ac,z^\Bc, \cdot)$, outputs $z^\Ac z^\Bc$. By definition,
	\begin{enumerate}
		\item $\pr{\Dc(U_{W_{I}^{\Ac,\pk}} ,U_{W_{I}^{\Bc,\pk} } , M^\pk_{\le I},S^\pk) =1} = \ex{U_{W_{I}^{\Ac,\pk}} \cdot U_{W_{I}^{\Bc,\pk}}} = \ex{W_{I}^{\Ac,\pk} \cdot W_{I}^{\Bc,\pk}}$, and 
		
		\item $\pr{\Dc(E_{I+1}^{\Ac,\pk}, Z_{I}^{\Bc,\pk} , M^\pk_{\le I},S^\pk)=1} = \ex{E_{I+1}^{\Ac,\pk}\cdot Z_{I}^{\Bc,\pk}}$.
		
	\end{enumerate}
	Hence, the proof  follows by   \cref{eq:Decor}.
\end{proof}

\paragraph{Proving \cref{claim:AttackOpurunityRes3}.}
\begin{proof}[Proof of \cref{claim:AttackOpurunityRes3}]   
	Consider the algorithm $\Dc$ that on input $(z^\Ac,z^\Bc, ( m_{\le I},\rdm))$: (1) computes $(\cdot, f^\Bc)=\Fc(m_{\le I};\rdm)$, (2) samples $u\la U_{f^\Bc}$, (3) outputs $z^\Ac\cdot u$.  By definition,
	\begin{enumerate}
		\item $\pr{\Dc(U_{W_{I}^{\Ac,\pk}} ,U_{W_{I}^{\Bc,\pk} } , M^\pk_{\le I},S^\pk) =1} = \ex{U_{W_{I}^{\Ac,\pk}}\cdot U_{F_{I}^{\Bc,\pk}}} = \ex{W_{I}^{\Ac,\pk} \cdot F_{I}^{\Bc,\pk}} $, and 
		
		\item $\pr{\Dc(E_{I+1}^{\Ac,\pk}, Z_{ I}^{\Bc,\pk} , M^\pk_{\le I},S^\pk)=1} = \ex{E_{I+1}^{\Ac,\pk}\cdot U_{F_{I}^{\Bc,\pk}}} = \ex{E_{I+1}^{\Ac,\pk}\cdot F_{I}^{\Bc,\pk}}$.
		
	\end{enumerate}
	Hence, also in this case the proof  follows by   \cref{eq:Decor}.
\end{proof}

\paragraph{Proving \cref{claim:AttackOpurunityRes2}.}
\begin{proof}[Proof of \cref{claim:AttackOpurunityRes2}]   
	
	Since $\size{W_{I}^{\Ac,\pk}}\le 1$, it suffices to prove $\ex{ \size{W_{I}^{\Bc,\pk} - F_I^{\Bc,\pk}}}\le 2\rho$. We show that if $\ex{ \size{W_I^{\Bc,\pk} - F_I^{\Bc,\pk}}}> 2\rho$, then there exists a distinguisher with advantage greater than $\rho$ for either the real outputs of $\hPi$ and the emulated outputs of $\Dcr$, or, the real outputs of $\tPi$ and the emulated outputs of $\Fc$, in contradiction with the assumed properties of $\Dcr$ and $\Fc$.

	Consider  algorithm \Dc that  on input $(z^\Ac,z^\Bc, m_{\le i},\rdm)$ acts as follows: (1) samples $\wh{\rdm}\la \wh{S}^\pk$,  (2) computes $(\cdot, f^\Bc)=\Fc(m_{\le i};\rdm)$ and $(\cdot, w^\Bc)=\Dcr(m_{\le i}, \rdm;\wh{s})$, (3) outputs $z^\Bc$ if $w^\Bc\ge f^\Bc$, and $1-z^\Bc$ otherwise.  We compute the difference in probability that $\Dc$ outputs $1$ given a sample from $\Dcr(M^\pk_{\le I})$ or a sample from $\Fc(M^\pk_{\le I})$   (we omit the superscript $\kappa$ and subscript $I$ below to reduce clutter)
	\begin{align*}
	\lefteqn{\pr{\Dc(U_{W_I^{\Ac,\pk}},U_{W_I^{\Bc,\pk}} , M^\pk_{\le I},S^\pk)=1 } - \pr{\Dc(U_{F_I^{\Ac,\pk}},U_{F_I^{\Bc,\pk}} , M^\pk_{\le I},S^\pk)=1 } }\\
	&  = \ex{U_{W  ^{\Bc}}\mid W  ^{\Bc}\ge F  ^{\Bc}}\cdot\pr{W  ^{\Bc}\ge F  ^{\Bc}} + \ex{1-U_{W  ^{\Bc}}\mid W  ^{\Bc}< F  ^{\Bc}}\cdot\pr{W  ^{\Bc}< F  ^{\Bc}} \\
	& \qquad- \ex{U_{F  ^{\Bc}}\mid W  ^{\Bc}\ge F  ^{\Bc}}\cdot\pr{ W  ^{\Bc}\ge F  ^{\Bc}} - \ex{1-U_{F  ^{\Bc}}\mid W  ^{\Bc}< F  ^{\Bc}}\cdot\pr{  W  ^{\Bc}< F  ^{\Bc}}\\
	&  = \ex{ W  ^{\Bc} \mid W  ^{\Bc}\ge F  ^{\Bc}}\cdot\pr{W  ^{\Bc}\ge F  ^{\Bc}} - \ex{ W  ^{\Bc} \mid W  ^{\Bc}< F  ^{\Bc}}\pr{ W  ^{\Bc}< F  ^{\Bc}} \\
	& \qquad- \ex{ F  ^{\Bc} \mid W  ^{\Bc}\ge F  ^{\Bc}}\cdot\pr{W  ^{\Bc}\ge F  ^{\Bc}} +\ex{  F  ^{\Bc} \mid W  ^{\Bc}< F  ^{\Bc}}\cdot \pr{ W  ^{\Bc}< F  ^{\Bc}}
	\\
	&  = \ex{ W  ^{\Bc}-  F  ^{\Bc} \mid W  ^{\Bc}\ge F  ^{\Bc}}\cdot\pr{W  ^{\Bc}\ge F  ^{\Bc}}   + \ex{ -W  ^{\Bc} + F  ^{\Bc} \mid W  ^{\Bc}< F  ^{\Bc}}\pr{ W  ^{\Bc}< F  ^{\Bc}} \\
	&= \ex{\size{W  ^{\Bc} - F  ^{\Bc}}} \\
	&> 2\rho.
	\end{align*}
	An averaging argument yields that  either $\Dc$  is a  distinguisher for $(U_{F_I^{\Ac,\pk}},U_{F_I^{\Bc,\pk}} , M^\pk_{\le I},S^\pk)$  and $(Z_I^{\Ac,\pk} ,Z_I^{\Bc,\pk}  , M^\pk_{\le I},S^\pk)$ with advantage greater than $\rho$, in contradiction with \cref{claim:goodForcaster}, or, $\Dc$  is a distinguisher for $(U_{W_I^{\Ac,\pk}},U_{W_I^{\Bc,\pk}} , M^\pk_{\le I},S^\pk)$  and $(E_I^{\Ac,\pk},Z^{\Bc,\pk}_I  , M^\pk_{\le I},S^\pk)$ with advantage greater than $\rho$, in contradiction with \cref{eq:Decor}. 
\end{proof}